\documentclass[a4paper,11pt]{article}

\usepackage{amsmath,amssymb}

\usepackage{graphicx}
\usepackage{latexsym}
\usepackage{makeidx}  

\newcommand{\nop}[1]{}

\def\cal{\mathcal}

\def\noi{\noindent}

\def\med{\medskip}

\newcommand\vsp{\vspace*{-0.3cm}}
\newcommand{\qed}{\hfill $\Box$}

\newenvironment{proof}{
\noindent{\bf\scshape Proof.}}
{\medskip
{\hfill $\Box$\smallskip}}  

\newtheorem{theorem}{Theorem}[section]{\bfseries}{\itshape}
\newtheorem{lemma}[theorem]{Lemma}{\bfseries}{\itshape}        
\newtheorem{corollary}[theorem]{Corollary}{\bfseries}{\itshape}     
\newtheorem{claim}[theorem]{Claim}{\bfseries}{\itshape}
{\bfseries}{\itshape}
\newtheorem{definition}[theorem]{Definition}{\bfseries}{\itshape}
\newtheorem{example}[theorem]{Example}{\bfseries}{\itshape}
\newtheorem{remark}[theorem]{Remark}{\bfseries}{\itshape}

\begin{document}

\title{Posimodular Function Optimization}

\author{
Toshimasa Ishii\thanks{Graduate School of Economics,
Hokkaido University,
Sapporo 060-0809, Japan 
({\tt ishii@econ.hokudai.ac.jp})} 
\and 
Kazuhisa Makino\thanks{Research Institute for Mathematical Sciences, 
Kyoto University, Kyoto 606-8502, Japan
({\tt makino@kurims.kyoto-u.ac.jp})}}

\nop{
\institute{Toshimasa Ishii \at
Graduate School of Economics,
Hokkaido University,
Sapporo 060-0809, Japan\\
\email{ishii@econ.hokudai.ac.jp} 
\and 
Kazuhisa Makino \at
Graduate School of Information Science and Technology, University of
Tokyo,
Tokyo 113-8656, Japan \\
\email{makino@mist.i.u-tokyo.ac.jp}}
}

\date{\empty}

\maketitle

\begin{abstract}
Given a posimodular function $f: 2^V \to \mathbb{R}$ on a finite set
 $V$,
we consider the problem of finding a nonempty subset $X$ of $V$ that
minimizes $f(X)$.
Posimodular functions often arise  in
 combinatorial
optimization such as undirected cut functions.
In this paper, we show that any algorithm for the problem requires 
  $\Omega(2^{\frac{n}{7.54}})$ oracle calls to $f$, where $n=|V|$.
It contrasts to the fact that the submodular function minimization, 
which is another generalization of cut functions, 
is polynomially solvable.

When  the range of a given posimodular function is  restricted to be $D=\{0,1,\ldots,d\}$ for some nonnegative integer $d$,
we show that  $\Omega(2^{\frac{d}{15.08}})$ oracle calls
are necessary, 
while we 
 propose
an $O(n^dT_f+n^{2d+1})$-time algorithm for the problem. 
Here,  $T_f$ denotes the time needed to evaluate the function value $f(X)$ for 
a given $X
\subseteq V$. 

We also consider the problem of maximizing a given posimodular
 function.
We 
 show that $\Omega(2^{n-1})$ oracle calls are necessary for
solving the  problem, 
and that the problem has
 time complexity  $\Theta(n^{d-1}T_f) $ when 
 $D=\{0,1,\ldots, d\}$ is the range of $f$ for 
some constant $d$. 
\end{abstract}

%

\med 


\vsp

\section{Introduction}\label{intro-sec}
\setcounter{equation}{0}
Let  $V$ denote a finite set with $n=|V|$. A set function $f: 2^V \to \mathbb{R}$
is called {\em posimodular} if
\begin{equation}\label{posi:eq}
 f(X)+f(Y)\geq f(X \setminus Y)+f(Y \setminus X) 
\end{equation}
 for all $X,Y \subseteq V$,
where $\mathbb{R}$ denotes the set of all  reals.
Posimodularity is one of the most fundamental and important properties in combinatorial optimization
 \cite{fuji00laminar,TM10posi,naga10mindegree,NI00poly,NSI01aug,SMNF09}.
Typically, it is a key for efficient solvability of  
undirected  network optimization and the related problems, 
since cut functions for undirected networks are posimodular.
Note that  cut functions for directed networks are not posimodular.
We can observe that posimodularity helps to create
complexity gaps for a number of network optimization problems, in the sense that 
the undirected versions  can be solved faster than
the directed versions.
For example,  the local edge-connectivity augmentation problem
in undirected networks is polynomially solvable, 
but  the problem in directed networks is NP-hard \cite{frank}.
As for the source location problem
with uniform demands or with uniform costs, 
the undirected versions can be solved in polynomial time
\cite{AIMF02locating,tamura2},
while the directed versions are NP-hard \cite{IMAHIF03}.
More generally, the currently fastest algorithm for 
minimizing a submodular and posimodular function
achieves $O(n^3T_f)$ time \cite{ni98note},
while the one for minimizing a submodular function
achieves $O(n^5T_f+n^6)$ time \cite{orlin09faster}, where
a set function $f: 2^V \to \mathbb{R}$
is called {\em submodular} if
\begin{equation}\label{sub:eq}
 f(X)+f(Y)\geq f(X \cap Y)+f(X \cup Y) 
\end{equation}
 for all $X,Y \subseteq V$, and 
 $T_f$ denotes the time needed to evaluate the function value $f(X)$ for 
a given $X
\subseteq V$. 
One of the reasons for these phenomena is based on  the following two structural properties on posimodular functions. 

A subset $X$ of $V$ is called {\em extreme} if
every nonempty proper subset $Y$ of $X$ satisfies $f(Y)> f(X)$. 
It is known that  the family ${\cal X}(f)$ of  extreme sets is {\em laminar} 
(i.e.,  
 every two members $X$ and $Y$ in  ${\cal X}(f)$  satisfy $X \cap Y=\emptyset$, $X\subseteq
Y$, or $X \supseteq Y$), when  $f$ is posimodular. 
Note that if $X,Y \in {\cal X}(f)$ would satisfy
$X \cap Y, X\setminus Y, Y\setminus X\neq \emptyset$,
then we have $f(X)+f(Y)\geq f(X \setminus Y)+f(Y \setminus X)> f(X)+f(Y)$, 
a contradiction.
The family ${\cal X}(f)$ of extreme sets
for an undirected cut function $f$ 
represents
the connectivity structure of a given network and helps to
design many efficient network algorithms
\cite{lawler73cutsets,WN87}.
For example, the undirected source location problem
with uniform demands can be solved in $O(n)$ time, if the
family  ${\cal X}(f)$ is known in advance, 
where $n$ corresponds to the number of vertices in the network \cite{naga04graph}. 
In fact, ${\cal X}(f)$ can be computed  in $O(n(m+n \log n))$ time for any undirected cut function  \cite{naga04graph}, 
where $m$ denotes the number of edges in the network. 
We note that ${\cal X}(f)$ can be found in $O(n^3T_f)$ time if $f$ is posimodular and submodular \cite{naga10mindegree}. 

The other structural property is for solid sets.
For an element $v \in V$,
 a subset $X$ of $V$ is called  {\em $v$-solid set}
if $v \in X$ and
every nonempty proper subset $Y$ of $X$ that contains  $v$
satisfies $f(Y)> f(X)$. 
Let ${\cal S}(f)$ denote the family of all solid sets, i.e.,  
${\cal S}(f) =\bigcup_{v \in V}\{ v$-solid $X\}$. 
It is known   \cite{SMNF09} that the
family  ${\cal S}(f)$ forms a tree hypergraph if $f$ is posimodular.
Similarly to the previous case for  ${\cal X}(f)$, 
if a host tree $T$ of ${\cal S}(f)$ is known in advance, 
this structure enables us to construct
a polynomial time algorithm for
the minimum transversal problem for posimodular functions $f$,  
which is 
an extension of the undirected source location problem
with uniform costs 
\cite{tamura2}
and the undirected external network problem \cite{heuvel05external}.  
If $f$ is in addition submodular,  a host tree $T$ can be computed in polynomial time. 

We here remark that 
these structural  properties on 
${\cal X}(f)$ and ${\cal S}(f)$ follow from the posimodularity of $f$, 
and that the submodularity is  needed to derive such structures efficiently, 
more precisely, 
the submodularity is assumed due to the property that 
$\min\{f(X)\mid \emptyset \neq X \subseteq V\}$ can be
computed in polynomial time.
%

On the other hand, to our best knowledge,
all the previous results for the posimodular optimization also make use of the submodularity or symmetricity, 
since undirected cut functions,   the most representative posimodular functions, are also submodular and symmetric.
Here a set function $f: 2^V \to \mathbb{R}$ is called {\em symmetric} if
$f(X)=f(V\setminus X)$ holds for any $X \subseteq V$.   
We note that a function is symmetric posimodular if and only if it is symmetric submodular, 
since the symmetricity of $f$ implies that 
$f(X)+f(Y)=f(V \setminus X)+f(Y)$ and
$f(X\setminus Y)+f(Y \setminus X)=f((V\setminus X)\cup Y)+f((V\setminus X)\cap Y)$.


In this paper, we focus on the posimodular  function minimization defined as follows. 
\begin{equation}\label{posi:prob}
\begin{array}{ll}
\multicolumn{2}{l}{\hspace*{-.10cm} 
{\mbox {\sc Posimodular Function Minimization}}}
\\[.08cm] 
\hspace*{-.10cm}
\mbox{Input:}& \mbox{A posimodular function}\,\,f: 2^V \to \mathbb{R},
 \\[.08cm]
\hspace*{-.10cm}
\mbox{Output:}& \mbox{A nonempty subset }X^* \mbox{ of $V$ such that }f(X^*)=\min_{X \subseteq V: X\not=\emptyset} f(X). 
\end{array} 
\end{equation}

\noindent
Here an input $f$ is given by an oracle that answers $f(X)$ for a given subset $X$ of $V$, and we assume that the optimal value $f(X^*)$ is also output. 
\noindent
The problem was posed as an open problem
on the Egres open problem
list \cite{egresopen}
in 2010,
as the negamodular function maximization,
where a set function $f$ is {\em negamodular}, if $-f$ is posimodular.
We also consider the posimodular function maximization, as the submodular function maximization has been intensively studied in recent years.

\subsection*{Our Contributions}

The main results obtained in this paper can be summarized as follows.

\begin{enumerate}
 \item  We show that   any  algorithm for the posimodular function minimization requires 
 $\Omega(2^{\frac{n}{7.54}})$ oracle calls.

 \item For a nonnegative integer $d$, let $D=\{0,1,\ldots, d\}$ denote the range of $f$, i.e., 
 $f: 2^V \to D$. 
Then we show that   $\Omega(2^{\frac{d}{15.08}})$ oracle calls
are necessary for 
the posimodular function minimization, 
while we 
 propose
an $O(n^dT_f+n^{2d+1})$-time algorithm for the problem. 
Also, as its byproduct, 
 the family ${\cal X}(f)$ of all extreme sets can be computed in 
$O(n^{d}T_f+n^{2d+1})$ time.
Furthermore, we show that all optimal solutions
can be generated with $O(nT_f)$ delay after 
generating
all locally minimal optimal solutions in $O(n^{d}T_f+n^{2d+1})$ time.

\item We show that the posimodular function maximization
requires 
 $\Omega(2^{n-1})$ oracle calls,
and that the problem has
 time complexity  $\Theta(n^{d-1}T_f) $ when 
 $D=\{0,1,\ldots, d\}$ is the range of $f$ 
for some constant $d$. 
\end{enumerate}

The first result contrasts to the submodular function minimization, which can be solved in polynomial time, and the second result implies the polynomiality for the posimodular function minimization if the range is bounded. The last result shows that the posimodular function maximization is also intractable.

The rest of this paper is organized as follows. 
 Section~\ref{preliminaries-sec} 
 presents basic definitions and preparatory properties on 
posimodular functions.  
In Section~\ref{hardness-sec}, 
we show the 
hardness results for the posimodular function minimization. 
In Section~\ref{dconst-sec}, 
we propose an $O(n^dT_f+n^{2d+1})$-time
algorithm for the posimodular function minimization
when $D$ is the range of $f$.  We also consider the problems for computing all extreme sets and all optimal solutions.  
Section~\ref{max-sec} treats the posimodular function maximization. 


\section{Preliminaries}\label{preliminaries-sec}
\setcounter{equation}{0}

Let $V$ be a finite set with  $n=|V|$.
For two subsets $X,Y$ of $V$, we say that {\em 
$X$ and $Y$ intersect each other} if each of $X\setminus Y$, $Y
\setminus
X$, and
$X\cap Y$ is nonempty. 
\nop{
A set function $f: 2^V \to \mathbb{R}$ is called {\em posimodular}
if 
\begin{equation}\label{posi:eq}
 f(X)+f(Y)\geq f(X \setminus Y)+f(Y \setminus X) 
\end{equation}
for every two subsets $X,Y$ of $V$.}
Let  $f: 2^V \to \mathbb{R}$ be a posimodular function.
Notice that any posimodular function $f$ satisfies
\begin{equation}
\label{eq--1}
f(X)\geq f(\emptyset) \mbox{ for all } X \subseteq V, 
\end{equation}
 since $f(X)+f(X)\geq f(\emptyset)+f(\emptyset)$.
Throughout the paper, we assume that $f(\emptyset)=0$, since otherwise,  
we can replace $f(X)$ by $f(X)-f(\emptyset)$ for all $X\subseteq V$. 

We here show a preparatory lemma for posimodular functions. 

\nop{

\begin{lemma}\label{max':lem}
Let  $f: 2^V \to \mathbb{R}$ be a posimodular function, and 
let $X$, $Y$, and $Z$ be three pairwise disjoint subsets of $V$. 
\begin{description}
\setlength{\itemindent}{-0.7cm}
\setlength{\parskip}{0cm} 
\item[$(${\it i}$)$]  If $f(X)\geq f(X \cup Y)$, then $f(Y \cup Z)\geq f(Z)$.
\item[$(${\it ii}$)$]  If $f(X)\geq f(Y\cup Z)$,
then   $f(X\cup Y)\geq f(Z)$. 
\end{description}
\end{lemma}
\begin{proof}
By (\ref{posi:eq}), we have
$f(X \cup Y)+f(Y \cup Z)\geq f(X)+f(Z)$, which proves (i) and (ii).
\end{proof}

\noindent
The next properties follow from Lemma~\ref{max':lem}.
}

\begin{lemma}\label{max:lem}
For  a posimodular function  $f: 2^V \to \mathbb{R}$,  
let  $T$ be a  subset of $V$ with
 $f(T)=\max\{f(X)\mid X \subseteq V\}$. 
For a  nonempty proper subset $U$ of $V$, the following two properties hold.
\begin{description}
\setlength{\itemindent}{-0.7cm}
\setlength{\parskip}{0cm} 
\item[$(${\it i}$)$]  If  $U \cap T=\emptyset$, then we have  $f(U)\geq f(\{v\})$ for any $ v \in U$. 
\item[$(${\it ii}$)$]  If  $U \supseteq T$, then we have  $f(U)\geq f(\{v\})$ for any $ v \not\in U$. 
\end{description}
\end{lemma}
\begin{proof}
If $T=V$, then the statements $({\it i})$ and $({\it ii})$ of the lemma clearly hold, 
since no nonempty proper subset $U$ of $V$ satisfies $U \cap T=\emptyset$
or $U \supseteq T$. On the other hand, if $T=\emptyset$, 
then we have 
$f(X)=0$ for all $X$ by (\ref{eq--1}) and the assumption on $f$. 
This again implies the statements of the lemma. 
We therefore assume that $T$ is a nonempty proper subset of $V$. 

For a nonempty subset $U$ with $U \cap T =\emptyset$, let $v$ be an element in $U$.
Then by (\ref{posi:eq}), we have $f(U)+f(T \cup (U \setminus \{v\})) \geq f(T) +f(\{v\})$. 
Since $T$ is a maximizer of $f$,  $f(U)  \geq f(\{v\}) $ holds, which proves  $({\it i})$ of the lemma. 
For a proper subset $U$ with $U \supseteq T$, let $v$ be an element in $V\setminus U$.
Then by (\ref{posi:eq}), we have $f(U)+f( (U \setminus T) \cup \{v\}) \geq f(T) +f(\{v\})$. 
Since $T$ is a maximizer of $f$,  $f(U)  \geq f(\{v\}) $ holds, which proves $({\it ii})$ of the lemma. 
\end{proof}


In this paper, we sometimes utilize a  Boolean function
 $\varphi:\{0,1\}^V \to \{0,1\}$. 
Let $x_v$ ($v \in V)$ be a Boolean variable, and a {\em literal} means a Boolean variable $x_v$ or its complement $\overline{x}_v$. 
A disjunction of literals  $c = \bigvee_{v \in P(c)}x_v \vee
\bigvee_{i \in N(c)}\overline{x}_v$ is called a {\em clause} if  $P(c) \cap N(c)=\emptyset$, 
 and  a {\em conjunctive normal form} (CNF, in short) is 
a conjunction of clauses. 
A  CNF is called {\em Horn}, {\em definite Horn},  and {\em dual Horn}
if each clause has at most one positive literal,
 exactly one positive literal, and
at most one negative literal, respectively.

\section{Hardness of the posimodular function minimization}\label{hardness-sec}
\setcounter{equation}{0}
In this section, 
we analyze the number of oracle calls
necessary to solve the posimodular function minimization.

Let
$g :2^V \to \mathbb{R}_+$ be a function defined by
$g(X)=|X|$ if $X \neq \emptyset$, and $g(\emptyset)=0$.
Clearly, $g$ is posimodular, since $g$ is monotone, i.e.,  
$g(X)\geq g(Y)$ holds for all two subsets $X$ and $Y$ of $V$ with $X \supseteq Y$. 
For a positive integer $k$ with $k \leq n/2$, 
let $S$ be a subset of $V$ of size $|S| =2k$.
Define a function  $g_S :2^V \to \mathbb{R}_+$ by
\[
g_S(X)=
\left\{ \begin{array}{ll}
2k-|X| & \mbox{ if } X \subseteq S \mbox{ and } |X|\geq k+1, \\
g(X)&\mbox{ otherwise}. 
\end{array} \right.
\]
\noindent
We can see that $g_S$ is a posimodular function close to $g$. 

\begin{claim}
\label{claim--1}
 $g_S$ is posimodular.
\end{claim}
\begin{proof}
Note first  that $g_S(X)\leq |X|$ for all $X\subseteq V$, 
since $|X|-(2k-|X|)\geq 0$ if  $|X|\geq k+1$.
Let $X$ and $Y$ be two subsets of $V$ with $X \cap Y \not=\emptyset$. 
We separately consider the following two cases. 

If at least one of $X$ and $Y$ has the identical function values for $g_S$ and $g$, say $g_S(X)=g(X)$,  
then we have $g_S(X)-g_S(X\setminus Y)\geq |X \cap Y|$. 
If $g_S(Y)=g(Y)$ is also satisfied, then we obtain 
$g_S(Y)-g_S(Y \setminus X) \geq |X \cap Y|$, and hence the posimodular inequality (\ref{posi:eq}) holds for such $X$ and $Y$. 
On the other hand, if $g_S(Y)\not=g(Y)$, i.e., $Y \subseteq S$ and $|Y|\geq k+1$,  
then  we have $g_S(Y)-g_S(Y \setminus X) \geq -|X\cap Y|$, which again implies the posimodular inequality  (\ref{posi:eq}). 

If $g_S(X)\not=g(X)$ and $g_S(Y)\not=g(Y)$ are satisfied, then 
 we have   $g_S(X)=2k-|X|$ and $g_S(Y)=2k-|Y|$. Since  
$|X\setminus Y|, |Y\setminus X|\leq k$,  we also have 
$g_S(X\setminus Y)=|X\setminus Y|$
and
$g_S(Y\setminus X)=|Y\setminus X|$.
Hence, it holds that 
\begin{eqnarray*}
g_S(X)+g_S(Y)-(g_S(X\setminus Y)+g_S(Y \setminus X))=4k-2|X\cup Y| \,\geq 0,
\end{eqnarray*}
where the last inequality follows from $X\cup Y\subseteq S$ and $|S|= 2k$.
Therefore the posimodular inequality (\ref{posi:eq})  holds. 
\end{proof}

Let   ${\cal G}=\{g\} \cup \{g_S \mid S\subseteq V, |S|=2k\}$. 
We below show that exponential oracles is necessary to distinguish among posimodular functions in $\cal G$.

Let ${\cal S}=\{S \subseteq V \mid |S|=2k\}$ and  ${\cal T}=\{T \subseteq V \mid k+1 \leq |T| \leq 2k\}$. 
Consider  the following integer programming
problem:
\begin{equation}\label{cover:prob}
\begin{array}{lll}
 \mbox{minimize} & \sum_{T \in {\cal T}} z_T& \\[.08cm] 
\mbox{subject to}& \sum_{T \in {\cal T}: T\subseteq S} z_T \geq 1
&  \mbox{for each }S \in {\cal S},
 \\[.08cm]
&  z_T \in \{0,1\} & \mbox{for each }T \in {\cal T}. 
\end{array} 
\end{equation}

\noindent
Note that any posimodular function $f$ in ${\cal G}$ satisfies $f(X)=g(X)$ if  $|X|\leq k$ or $|X| \geq 2k+1$. 
Oracle calls for such sets $X$ do not help to distinguish  among posimodular functions in $\cal G$. 
Therefore, 
we can restrict our attention to subsets $T$ in ${\cal T}$ for oracle calls.

\begin{lemma}
\label{lemma--1}
Let $q_k$ denote the optimal value for $(\ref{cover:prob})$.
Then  at least $q_k$ oracle calls is necessary to distinguish among posimodular functions in $\cal G$.
\end{lemma}

\begin{proof}
Assume to the contrary that there exists an algorithm $A$ 
which distinguishes by  at most $q_k-1$ oracle calls among posimodular functions in $\cal G$.
Let $\cal X$ denote the family of subsets of $V$ which are called by $A$ if a posimodular function $g$ is an input of $A$. Since $|{\cal X}| \leq q_k-1$, we have a subset $S$ in $\cal S$ such that no $X \in {\cal X}$ satisfies 
$X \subseteq S$ and $|X| \geq k+1$. 
This means that $g_S(X)=g(X)$ for all $X \in {\cal X}$, which contradicts that $A$ distinguishes between $g$ and $g_S$. 
\end{proof}

It follows from Lemma \ref{lemma--1} that $q_k$ oracle calls are required for the posimodular function minimization. 
We now analyze the optimal value $q_k$ for $(\ref{cover:prob})$.

\begin{lemma}
\label{lemma--2}
Let $q_k$ denote the optimal value for $(\ref{cover:prob})$.
Then we have  $q_k \geq {n \choose k+1}/{2k \choose k+1}$. 
\end{lemma}

\begin{proof}
Consider the linear programming relaxation for Problem (\ref{cover:prob}) which is obtained by replacing 
 each binary constraint $z_T \in \{0,1\}$ by $z_T \geq 0$: 

\begin{equation}\label{cover:probLP}
\begin{array}{lll}
 \mbox{minimize} & \sum_{T \in {\cal T}} z_T& \\[.08cm] 
\mbox{subject to}& \sum_{T \in {\cal T}: T\subseteq S} z_T \geq 1
&  \mbox{for each }S \in {\cal S},
 \\[.08cm]
&  z_T \geq 0 & \mbox{for each }T \in {\cal T}. 
\end{array} 
\end{equation}
Define a vector $z^* \in \mathbb{R}^{\cal T}$ by 
$z^*_T=1/{2k \choose k+1}$ if $|T|=k+1$, and $0$ otherwise. 
Note that $z^*$ is feasible to (\ref{cover:probLP}), and the objective value is 
\begin{eqnarray}
\label{eq--3}
\sum_{T \in {\cal T}} z^*_T&=&\frac{{n \choose k+1}}{{2k \choose k+1}}. 
\end{eqnarray}
Moreover, we show that it is optimal to (\ref{cover:probLP}). 

Define $y \in \mathbb{R}^{\cal S}$ by $y_S=1/{n-(k+1) \choose k-1}$ for all $S \in {\cal S}$. 
Then this $y$ is feasible to the dual problem  of  (\ref{cover:probLP}), and the objective value is 
\begin{eqnarray}
\label{eq--4}
\sum_{S \in {\cal S}} y_S&=&\frac{{n \choose 2k}}{{n-(k+1) \choose k-1}} \,=\,\frac{{n \choose k+1}}{{2k \choose k+1}}. 
\end{eqnarray}
By (\ref{eq--3}) and (\ref{eq--4}), $z^*$ is an optimal solution of (\ref{cover:probLP}). 
Since it is a relaxation of the minimization problem, we have $q_k \geq {n \choose k+1}/{2k \choose k+1}$. 
\end{proof}

For $k \geq 2$, we note that 
\begin{eqnarray}
  \frac{{n \choose k+1}}{{2k \choose k+1} } & = &  \frac{n! (k-1)!}{(2k)!(n-k-1)!} \nonumber \\[.2cm]
& \geq &
\frac{2\pi}{e^2}\cdot\frac{n^{n+1/2}(k-1)^{k-1/2}}{(2k)^{2k+1/2}(n-k-1)^{n-k-1/2}}\nonumber \\[.2cm]
& \geq & \frac{2\pi}{e^2}\cdot  \Bigl(\frac{n}{2k}\Bigr)^{k+1}\!\!\cdot
\Bigl(\frac{k-1}{2k}\Bigr)^{k-1/2} \nonumber \\[.2cm]
& \geq & \frac{2\pi}{e^2}\cdot  \Bigl(\frac{n}{4k} \Bigr)^{k+1}. \label{eq--5} 
\end{eqnarray} 
Here the second, third, and fourth inequalities respectively follow from
Stirling's inequalities
$\sqrt{2\pi}n^{n+1/2}e^{-n} \leq n! \leq e n^{n+1/2}e^{-n}$,
$n \geq n-k-1$, and 
$(1-\frac{1}{k})^{k-1/2}\geq \frac{1}{2\sqrt{2}}$ for $k\geq 2$. 
By setting   $n=\lceil 4ek \rceil$, 
we obtain that (\ref{eq--5}) is 
 $\Omega(e^{\frac{n}{4e}})=\Omega(2^{\frac{n}{7.54}})$.

Thus, we have the following theorem.
\begin{theorem}\label{hardness1:th}
Any algorithm for the posimodular function minimization  requires 
 $\Omega(2^{\frac{n}{7.54}})$ oracle calls.
 \end{theorem}

\nop{
Also, we can derive the following inapproximability result as a corollary of Theorem \ref{hardness1:th}. 
\begin{corollary}
Any $2$-approximation algorithm for the posimodular function minimization  requires 
 $\Omega(2^{\frac{n}{7.54}})$ oracle calls.
 \end{corollary}
}

Let us next consider the case in which the range of $f$ is bounded by $D=\{0,1,\ldots,d\}$ for some nonnegative integer $d$. 
We show the exponential lower bound  in a similar way to the proof of 
Theorem~\ref{hardness1:th}.

Let $T$ be a subset of $V$ with
 $|T|=\lfloor d/ 2 \rfloor$. 
Define 
$g :2^V \to D$  by
\[
g(X)=
\left\{ \begin{array}{ll}
0 &  \mbox{ if }X=\emptyset, \\
|X| & \mbox{ if }\emptyset \neq X \subseteq T,\\ 
|T|+|T \cap X| & \mbox{ otherwise}.
\end{array} \right.
\]
\noindent
For a positive integer $k$ with $2k \leq |T|$, 
let $S$ be a subset of $T$ with $|S| =2k$. 
Define a function  $g_S :2^V \to D$ by
\[
g_S(X)=
\left\{ \begin{array}{ll}
2k-|X| & \mbox{if } X \subseteq S \mbox{ and }|X|\geq k+1, \\
g(X) & \mbox{otherwise}. 
\end{array} \right.
\]

\noindent
The monotonicity of $g$ implies that $g$ is posimodular.
The posimodularity of $g_S$ can be shown as follows.

Let $X$ and $Y$ be two subsets of $V$.
If both $X$ and $Y$ are subsets of $T$, then 
the posimodular inequality (\ref{posi:eq}) follows from Claim \ref{claim--1}.
We therefore assume that  $X \setminus T \neq \emptyset$. 
Then $g_S(X)=|T|+|T \cap X|$ holds.
Note that $|T|+|T \cap Z|\geq g_S(Z)$ holds for all $Z \subseteq T$.
Thus we have 
\begin{eqnarray}
g_S(X)- g_S(X\setminus Y)&\geq &  |T \cap X|-|T \cap (X
\setminus Y)|\nonumber\\
&=&|T \cap X \cap Y| \,\,\,(\geq 0). \label{eq--7}
\end{eqnarray}
If $Y \setminus T \neq \emptyset$ is also satisfied, then
we have 
 $g_S(Y) \geq g_S(Y\setminus X)$,
from which the posimodular inequality (\ref{posi:eq}) holds. 
On the other hand, if $Y \subseteq  T$, 
then we have $g_S(Y)-g_S(Y\setminus X) \geq -|X \cap Y|$ by $Y \setminus X \subseteq T$. 
Moreover, (\ref{eq--7}) implies $g_S(X)-g_S(X\setminus Y) \geq |X\cap Y|$ by $X\cap Y \subseteq T$, 
and hence we obtain (\ref{posi:eq}). 

If $k \geq 2$ and $|T|\approx  4ek$ (and hence $d\approx 8ek$), then 
by applying an argument similar to Lemmas \ref{lemma--1} and \ref{lemma--2}, we have the following result.

\begin{theorem}
Assume that a given  posimodular function has range $D$. 
Then the posimodular function minimization  requires 
 $\Omega(2^{\frac{d}{15.08}})$ oracle calls.
 \end{theorem}

\section{Polynomial time algorithm for posimodular function minimization when $d$ is  a constant}\label{dconst-sec}
\setcounter{equation}{0}
In this section, we show that the posimodular function minimization can be solved in polynomial time if 
an input posimodular function is restricted to be $f: 2^V \to \{0,1,\ldots ,d \}$ for some constant $d$. 
We first show that for $d \leq 3$,  the posimodular function minimization can be solved efficiently by  repeatedly contracting semi-extreme sets, and then provides an $O(n^dT_f+n^{2d+1})$-time algorithm for general $d$. 

In this section, 
an optimal solution to the posimodular function minimization (\ref{posi:prob})
 is referred to as {\em a minimizer of $f$} (among nonempty subsets). 


\subsection{Case in which $d \in \{0,1,2,3\}$}

 
Let $f:2^V\to \{0,1,\dots , d\}$ be a function, and let $s$ be an element with $s \not\in V$. 
For a subset $S \subseteq V$, 
let $f':2^{ (V \setminus S)\cup \{s\}} \to \{0,1,\dots , d\}$ be a function defined by 
\[
f'(X)=
\left\{ \begin{array}{ll}
f(X) & \mbox{if } s \not\in X, \\
f((X\setminus \{s\}) \cup S) & \mbox{otherwise}. 
\end{array} \right.
\]
We say that  {\em the function $f'$ 
is obtained from $f$  by  contracting a subset $S$ of $V$
into an element $s$}.  
Notice that $f'$ is posimodular if it is obtained form 
 a posimodular function  by contraction. 
A nonempty subset $X$ of $V$ is called {\em semi-extreme} (w.r.t $f$) if
all nonempty subsets $Y$ of $X$ satisfy $f(Y)\geq f(X)$.
By the following lemma, we can contract any semi-extreme set
while keeping at least one minimizer of $f$.  

\begin{lemma}\label{extreme:lem}
Let $f$ be a posimodular function.  For any   semi-extreme set $X$,
there exists a minimizer  $Y$ of $f$ such that $Y \supseteq X$ or $X\cap Y =\emptyset$. 
\end{lemma}
\begin{proof}
Assume  that a minimizer  $Y$ of $f$ satisfies  $Y \not\supseteq X$ and $X\cap Y \not=\emptyset$. 
If $Y$ is a subset of $X$, then $X$ is also a minimizer of $f$ by  the semi-extremeness of $X$. 
On the other hand, if $Y$ intersects $X$, then it follows from  (\ref{posi:eq}) that 
$f(X)+f(Y)\geq f(X \setminus Y)+f(Y \setminus X)$.
Since $X$ is semi-extreme, we have $f(X)\leq f(X \setminus Y)$.
It follows that $f(Y)\geq f(Y \setminus X)$, which implies that 
$Y \setminus X$ is also a minimizer of $f$. 
\end{proof}

The following lemma indicates that we can obtain
a minimizer of $f$ after contracting a subset $X$ of $V$
with $|X|= 2$ at most $n$ times.

\begin{lemma}\label{d<=3:cl}
If   $d \leq 3$, then 
there exists a semi-extreme set $X$
 with $|X|= 2$,
or a minimizer  $Y$ of $f$ with $|Y| = 1$ or
$|Y| \geq n-1$.
\end{lemma}
\begin{proof}
Consider the case in which $n\geq 4$, since the lemma clearly holds for $n \leq 3$.
Assume to the contrary that no subset $X$ with $|X|= 2$ is semi-extreme and no subset $Y$ with $|Y|= 1$,  $n-1$ or $n$ is a minimizer of $f$. 
Let $X^*$ be a minimizer of $f$. 
Then by the assumption,  we have
\begin{equation}
\label{eq---1}
\left. \begin{array}{llll}
 f(Y) &\geq& f(X^*)+1  \,(\geq 1)  &\mbox{for all subsets } Y \mbox{ with } |Y|=1, n-1 \mbox{ or }n,  \\
 f(X) &\geq& f(X^*)+2  \,(\geq 2)  &\mbox{for all subsets } X \mbox{ with } |X|=2.  
\end{array} \right.
\end{equation}
This already proves   this lemma  for  $d=1$.

If $d=2$, then it follows from (\ref{eq---1})  that all subsets $X$ with $|X|=2$ satisfy $f(X)=2$.
Hence,  by Lemma~\ref{max:lem}, any
nonempty proper subset $Z$ of $V$ satisfies $f(Z)\geq \min\{f(v)\mid v \in V\}$.
This implies that some element of $V$ or  $V$ is a minimizer of $f$, which contradicts the assumption. 

For $d=3$,   we separately consider the cases in which 
the optimal value $f(X^*)$   is $0$, $1$, and  at least $2$.  

{\bf Case $f(X^*) \geq 2$}.   By  (\ref{eq---1}) we have $f(X)\geq 4$ for all subsets $X$ of $V$ with $|X|=2$,  
which contradicts the fact that $d=3$.

{\bf Case $f(X^*) =0$}. 
By the assumption, we have $|X^*|\geq 3$. Moreover,  
if $|X^*|\leq n-2$, then
there exists a subset $Z$ of $V$ such that  $|X^* \setminus Z|=|Z \setminus X^*|=2$.
By applying (\ref{posi:eq}) to $X^*$ and $Z$,
we have $3\geq f(X^*)+f(Z)\geq f(X^* \setminus Z)+f(Z \setminus X^*)$,
from which
$f(X^* \setminus Z)\leq 1$ or $f(Z \setminus X^*)\leq 1$.
Since this contradicts (\ref{eq---1}), we have $|X^*|\geq n-1$, which  again contradicts  (\ref{eq---1}).

{\bf Case $f(X^*)=1$}.
By  (\ref{eq---1}), all subsets $X$ of $V$ with $|X|=2$ satisfy $f(X)=3$. 
Similarly to the case of $d=2$, Lemma~\ref{max:lem} implies that any
nonempty proper subset $Z$ of $V$ satisfies $f(Z)\geq \min\{f(v)\mid v \in V\}$.
Hence some element of $V$ or  $V$ is a minimizer of $f$, which contradicts the assumption. 
\end{proof}

By the lemma, for $d\leq 3$, we first check function values $f(X)$ for all subsets $X$ with $|X|=1,2,n-1$, and $n$. If no subset $X$ with $|X|=2$ is semi-extreme, then we output a subset $X^*$ which satisfies $f(X^*)=\min_{X:|X|=1, n-1, {\rm or}\,\,n} f(X)$. 
Otherwise (i.e., if some $X$ with $|X|=2$ is semi-extreme),
 we consider the function $f'$ obtained from $f$ by contracting $X$  into a new element $x$, and check   $f'(X')$ for all subsets $X'$ with $|X'|=1,2,n-2$, and $n-1$.
Note that it is enough to check $f(X')$ for subsets $X'$ with $X' \ni x$ and $|X'|=2$, 
since the other $X'$ have been already checked during the first iteration.  
By repeating this procedure, we obtain a minimizer of $f$. 
Since the first iteration requires $O(n^2+n^2T_f)=O(n^2T_f)$ time and all the other iterations require 
 $O(n+nT_f)=O(nT_f)$, we have the following result.

\begin{theorem}\label{theremd<=3}
For  $d \leq 3$,  the posimodular function minimization can be solved in $O(n^2T_f)$ time. 
\end{theorem}

You might think that a similar property to Lemma \ref{d<=3:cl} holds  for a general $d$. 
However, the following instance indicates that this is not the case, since  no nontrivial semi-extreme set 
is small. In fact, the size of each nontrivial semi-extreme set is independent of $d$.  

\begin{example} \rm
Let $S$ be an arbitrary subset of $V$ with $|S|\geq 4$, and let $f: 2^V \to \{0,1,\ldots,7\}$
be a posimodular function defined as
\[
f(X)=\left\{ \begin{array}{ll}
0 & \mbox{ if }X=\emptyset,S \\
1 & \mbox{ if }X \subseteq S, |X|=1 \mbox{ or } |S|-1 \\
2 & \mbox{ if }X \subseteq S, 2 \leq |X| \leq |S|-2 \\
2 & \mbox{ if }X \cap S=\emptyset, |X|=1\\
3 & \mbox{ if }X \cap S=\emptyset, |X|\geq 2\\
4 & \mbox{ if }X \setminus S\neq \emptyset, |X \cap S|=1 \\
5 & \mbox{ if }X\setminus S\neq \emptyset, 2 \leq |X \cap S| \leq |S|-2\\
6 & \mbox{ if }X\setminus S\neq \emptyset,  |X \cap S| = |S|-1 \\
7 & \mbox{ if }X \setminus S\neq \emptyset, X \cap S = S. 
\end{array} \right.
\]
\noindent
We note that $\{S\} \cup \{ \{v\} \mid v \in X\} \cup \{S\setminus \{v\} \mid v \in S\}$ is the family of  all semi-extreme sets of $f$. Therefore, 
the size of each nontrivial semi-extreme set is either $|S|$ or $|S|-1$, which is 
independent of $d$. 
The posimodularity of  $f$ can be shown as follows.

Let $X$ and $Y$ be two subsets of $V$ with $X \cap Y\not=\emptyset$. 
If both $f(X)-f(X\setminus Y)$ and $f(Y)-f(Y\setminus X)$ are nonnegative, then 
(\ref{posi:eq}) clearly holds. 
We thus assume that at least one pair  $(Z_1,Z_2)$ of  $(X\setminus Y, X)$ and $(Y\setminus X,Y)$
satisfies one of the following conditions, where  $Z_1 \subseteq Z_2$. 

\medskip

(a) $Z_1 \subseteq S$, $|Z_1|=1$, and $Z_2=S$.

\smallskip

(b) $Z_1 \subseteq S$, $2 \leq |Z_1| \leq |S|-2$, and $Z_2=S$.

\smallskip

(c) $Z_1 \subseteq S$, $|Z_1|=|S|-1$, and $Z_2=S$.

\smallskip

(d) $Z_1,Z_2 \subseteq S$, $2 \leq |Z_1| \leq |S|-2$, and $|Z_2|=|S|-1$.

\medskip

If  $(Z_1,Z_2)=(X\setminus Y, X)$ satisfies
(a) (i.e., $X=S$ and $|X \setminus Y|=1$), then 
we have $f(X)=0$, $f(X\setminus Y)=1$, 
$|X \cap Y|=|S|-1$, and $(Y \setminus X) \cap S=\emptyset$. 
If $Y \setminus X=\emptyset$,
then we have $f(Y)=1$ and $f(Y\setminus X)=0$, which implies  (\ref{posi:eq}). 
On the other hand,  if  $Y \setminus X \neq \emptyset$, 
then we have $f(Y)=6$ and $f(Y\setminus X)\leq 3$,
 which again implies (\ref{posi:eq}).

If  $(Z_1,Z_2)=(X\setminus Y, X)$ satisfies
(b), then 
we have $f(X)=0$, $f(X\setminus Y)=2$, 
$2 \leq |X \cap Y| \leq |S|-2$, and $(Y \setminus X) \cap S=\emptyset$. 
Hence
if $Y \setminus X=\emptyset$,
then it holds that $f(Y)=2$ and $f(Y\setminus X)=0$. 
On the other hand, if  $Y \setminus X \neq \emptyset$, 
then we have $f(Y)=5$ and $f(Y\setminus X)\leq 3$. 
In either case,  (\ref{posi:eq})  is derived.

If  $(Z_1,Z_2)=(X\setminus Y, X)$ satisfies (c), 
then 
we have $f(X)=0$, $f(X\setminus Y)=1$, 
$|X \cap Y|=1$, and $(Y \setminus X) \cap S=\emptyset$.  
If $Y \setminus X=\emptyset$,
then it holds that $f(Y)=1$ and $f(Y\setminus X)=0$.
On the other hand,  if  $Y \setminus X \neq \emptyset$, 
then $f(Y)=4$ and $f(Y\setminus X)\leq 3$. 
In either case,  (\ref{posi:eq})  is derived.

If  $(Z_1,Z_2)=(X\setminus Y, X)$ satisfies
(d), then 
we have $f(X)=1$, $f(X\setminus Y)=2$, 
$1 \leq |X \cap Y| \leq |S|-3$, and $|(Y \setminus X) \cap S| \leq 1$. 
Hence
if $Y \setminus X=\emptyset$,
then $f(Y) \geq 1$ and $f(Y\setminus X)=0$.
If $Y \setminus X \neq \emptyset$ and $Y \setminus S=\emptyset$,
then  
 $f(Y) = 2$ and $f(Y\setminus X)=1$
by $Y \subseteq S$, $2\leq |Y|\leq |S|-2$, and
 $|Y \setminus X|=|(Y \setminus X) \cap S|=1$.  
If   $Y \setminus X \neq \emptyset$,  $Y \setminus S\neq \emptyset$,
and $(Y \setminus X) \cap S=\emptyset$,
then  $f(Y)\geq 4$ and $f(Y\setminus X)\leq 3$.
If   $Y \setminus X \neq \emptyset$,  $Y \setminus S\neq \emptyset$,
and $(Y \setminus X) \cap S \neq \emptyset$,
then  $f(Y)\geq 5$ and $f(Y\setminus X)= 4$
by $|Y\cap S|=|Y \cap X|+|(Y \setminus X) \cap S|\geq 2$
and $|(Y\setminus X)\cap S|=1$.
In either case,  (\ref{posi:eq})  is derived. 
\end{example}

\subsection{Case in which $d$ is general}
In this section, we propose an algorithm for the posimodular function minimization for general $d$.
Different from our algorithm for $d\leq 3$, it is not based on the contraction for semi-extreme sets. 
Instead, we focus on the following simple property derived from  posimodularity, 
and solve the problem by making use of dual Horn Boolean satisfiability problem. 
\begin{lemma}\label{key:lem}
For a nonnegative integer $d$, let $f:2^V\to \{0,1,\dots , d\}$ be a posimodular function.  
If there exist a subset $X$ of $V$ and an element $s\in V \setminus X$
such that 
\begin{equation}\label{key:eq}
f(X) \geq f(X\cup \{s\}),  
\end{equation}
then
any  subset $Y$ with $Y \cap X=\emptyset$
satisfies $f(Y)\geq f(Y \setminus \{s\})$.  
\end{lemma}
\begin{proof}
If $s \notin Y$, then we clearly have $f(Y)=f(Y\setminus\{s\})$.
On the other hand, if $v \in Y$, then 
by (\ref{posi:eq}), we have
$f(X \cup \{s\})+f(Y)\geq f(X)+f(Y\setminus \{s\})$, which proves the
 lemma. 
\end{proof}

Let us consider computing a  locally minimal  minimizer  $X^*$ of $f$. 
Here a subset $X^*$ is called {\em locally minimal} if $f(X^*) < f(X^* \setminus \{v\})$ holds   for any $v\in X^*$. 
We note that a locally minimal minimizer $X^*$ always exists if no singleton $\{v\}$ is a minimizer of $f$, and such an $X^*$ satisfies $|X^*|\geq 2$.  

Let $X$ be a  subset  of $V$,  and  let $s$ be an element  in $V \setminus X$
that satisfies (\ref{key:eq}). 
Then Lemma~\ref{key:lem} implies that
any locally minimal subset $X^*$ must satisfy $s \notin X^*$  whenever $X^*\cap X=\emptyset$. 
To represent it as a Boolean formula, let us introduce propositional variables $x_v$, $v\in V$, and we regard
 a Boolean vector $x \in \{0,1\}^V$  as a subset $S_x$ such that $S_x=\{v\in V\mid x_v=1\}$, i.e., 
$x$ is the characteristic vector of $S_x$. 
Then it can be represented as  
\begin{equation}\label{key2:eq}
\mbox{$x_s=0$
whenever $x_v =0$ for all $v \in X$},     
\end{equation}
which is equivalent to satisfying the following dual Horn clause 
\begin{equation}\label{key3:eq}
 \bigvee_{v \in X}x_v \vee \overline {x}_s. 
\end{equation}


If you have many  pairs of $X$ and $s$ that satisfy  (\ref{key:eq}), 
then their corresponding  rules 
 (\ref{key3:eq})  reduce the search space for finding a locally minimal minimizer of $f$. 
Note that the rules can be represented as a dual Horn CNF, 
and hence the satisfiability can be solved in linear time and
 all satisfiable assignments can be generated with linear delay 
  (i.e., the time interval between two consecutive output is bounded 
in linear time (in the input size))  \cite{pretolani93linear}.  
However,  the number of such pairs are in general exponential in $n$,
and hence we need to find a subfamily $\cal P$ of such pairs $(X,s)$ 
such that (1) the size $|{\cal P}|$ is polynomial in $n$ (for a constant $d$) and 
(2) the corresponding dual-Horn CNF has polynomially many satisfiable assignments.

\begin{definition}
Let $X$ be a subset of $V$ with $k=|X|$. 
We say that {\em $X$  
is reachable $($from $\emptyset)$}
if there exists a chain $X_0\,(=\emptyset) \subsetneq X_1 \subsetneq \dots \subsetneq X_k\,(=X)$
such that 
$f(X_i) > f(X_{i-1})$   for all $i=1,2,\ldots, k$,
and 
 {\em unreachable} otherwise.
\end{definition}

\noindent
By definition, $\emptyset$ is reachable.
An unreachable set $U$ is called {\em minimal} if  any proper subset of it is reachable.  
 Let ${\cal U}$ be the family of  minimal unreachable sets $U$. 
From the definition of reachability, we have the following lemma. 

\begin{lemma}\label{unreach1:lem}
For any minimal unreachable set $U \in {\cal U}$,
we have $f(U)\leq f(U \setminus \{u\})$ for all $u \in U$. 
\end{lemma}

\begin{proof}
By definition, $U\setminus \{u\}$ is reachable for all $u \in U$.
Hence, if $f(U)> f(U \setminus \{u \})$ holds for some $u  \in U$, 
then it turns out that $U$ is reachable, which is a contradiction. 
\end{proof}

\noindent
\begin{lemma}\label{opt:lem}
Let $X^*$ be a locally minimal subset  of a posimodular function $f$.  
Then the characteristic vector  of $X^*$ satisfies the dual Horn CNF   $\varphi_f$ defined by 
\begin{equation}\label{cnf:eq}
 \varphi_f=
 \bigwedge_{U \in {\cal U}} \bigwedge_{s \in U}
(\bigvee_{u \in U\setminus \{s\}}\!\!x_{u} \vee \overline {x}_s)
\end{equation}%
\end{lemma}

\begin{proof}
 Lemma \ref{unreach1:lem}, together with the discussion after  Lemma \ref{key:lem} implies the lemma.  
\end{proof}

\noindent
Based on the lemma, we have the following algorithm for the posimodular function minimization. 

\medskip

\noindent
{\bf Algorithm}  {\sc  MinPosimodular($f$)}

\smallskip

\noindent
{\bf Step 1}. Compute a singleton $\{v^*\}$ with minimum $f(v^*)$ (i.e.,  $f(v^*)=\min\{f(v)\mid v \in V\}$). 

\smallskip

\noindent
{\bf Step 2}. Compute a subset $S_{x^*}$ with minimum 
 $f(S_{x^*})$ among the sets $S_x$ such that 
$|S_x|\geq 2$ and $\varphi_f(x)=1$.

\smallskip

\noindent
{\bf Step 3}. Output  $\{v^*\}$, if $f(v^*) \leq f(S_{x^*})$, and $S_{x^*}$, otherwise. Halt. \qed


\medskip

In the remaining part of this section, we show that  $\varphi_f$ has polynomially many clauses and satisfiable assignments in $n$ (if $d$ is bounded by a constant). 


We first show basic facts for minimal unreachable sets, where a subset $I$ of $V$ is called {\em independent} of $\cal U$ if it contains no $U \in {\cal U}$.

\begin{lemma}\label{unreach2:lem}
For a posimodular function $f:2^V\to \{0,1, \dots , d\}$, we have the following three statements.
\begin{description}
\setlength{\itemindent}{-0.7cm}
\setlength{\parskip}{0cm} 
\item[$(${\it i}$)$] $1 \leq |U| \leq d+1$ holds for all $U  \in {\cal U}$.
\item[$(${\it ii}$)$] $|I|\leq d$ holds for all independent sets $I$ of ${\cal U}$. 
\item[$(${\it iii}$)$] If a singleton $\{u\}$ is contained in ${\cal U}$, then $f(u)=0$, and hence
$\{u\}$ is a minimizer of $f$. 
\end{description}
\end{lemma}
\begin{proof}
Since $\{0,1,\ldots,d\}$ is  the range of $f$, any reachable set $R$ has cardinality $|R|$ at most $d$.
This implies that $(i)$ and $(ii)$.
$(iii)$  follows from  $f(\emptyset)=0$ by our assumption.  
\end{proof}

Lemma \ref{unreach2:lem} $(i)$ implies that $|{\cal U}|=O({n \choose d+1})$ if $d< n/2$, and 
$O({n \choose n/2})$ otherwise. 
Hence we have  
\begin{eqnarray}
\label{eq-abc1}
|{\cal U}|&=&O(n^{d+1}/d).
\end{eqnarray}

Let us then analyze the number of satisfiable assignments of $\varphi_f$. 
In order to make the discussion simpler,  consider a definite Horn CNF $\varphi_f(\overline{x})$, 
where $\overline{x}$ denotes the complement  of $x$.
Notice that a subset $S_x$  with $\varphi_f(\overline{x})=1$
is a candidate of the complement of a locally minimal 
minimizer of $f$. 
For a definite Horn CNF $\varphi$ and a subset $T$ of $V$, 
the following algorithm called {\em forward chaining
 procedure {\rm (}FCP{\rm )}} has been proposed to compute satisfiable assignments of $\varphi$ \cite{arias09canonical,hammer93optimal}.


\medskip

\noindent
{\bf Procedure}  {\sc  FCP}($\varphi$; $T$)

\smallskip

\noindent
{\bf Step 0}.   Let $Q:=T$. 

\smallskip

\noindent
{\bf Step 1}. While there exists a clause $c$ in $\varphi$ such that $N(c) \subseteq Q$ and
$P(c)\cap Q=\emptyset$ do 

\hspace*{2.1cm}$Q:=Q \cup P(c)$.

\smallskip

\noindent
{\bf Step 2}. Output $Q$  as {\sc  FCP}($\varphi$; $T$), and halt.

\medskip
\noindent
It is not difficult to see that $T\subseteq {\sc  FCP}(\varphi; T)$ holds for any subset $T$, 
and  ${\sc  FCP}(\varphi; T) \subseteq {\sc  FCP}(\varphi; T')$ holds if $T \subseteq T'$. 
Moreover, for a definite Horn CNF $\varphi$, 
 it is known \cite{arias09canonical,hammer93optimal} that 
$T$ corresponds to a satisfiable assignment of $\varphi$ (i.e., the characteristic vector of  $T$ is a satisfiable assignment of $\varphi$)  if and only if $T= {\sc  FCP}(\varphi; T)$. 
This implies that 
for any subset $T$, ${\sc  FCP}(\varphi; T)$ corresponds to a satisfiable assignment of $\varphi$, and 
for any satisfiable assignment $\alpha$ of $\varphi$, 
there exists a subset $T$ such that ${\sc  FCP}(\varphi; T)$ corresponds to $\alpha$ (i.e., . $S_\alpha={\sc  FCP}(\varphi; T)$). 

We now claim that for any satisfiable assignment $\alpha$ of $\varphi_f(\overline{x})$, there exists a subset $T$ such that $|T|\leq d$ and $S_\alpha={\sc  FCP}(\varphi_f(\overline{x}); T)$, 
which implies 
the number of satisfiable assignments of $\varphi_f$ is bounded by $\sum_{i=0}^d{n \choose i}$. 

For a satisfiable assignment $\alpha$ of 
$\varphi(\overline{x})$, 
let 
${\cal U}_\alpha=\{U \subseteq {\cal U} \mid U \subseteq S_\alpha\}$,
and let $I \subseteq S_{\alpha}$ be an independent set of ${\cal U}_\alpha$ which is maximal in $S_{\alpha}$ (i.e.,  $I \cup \{v\}$ is dependent of  ${\cal U}_\alpha$ for all 
$v \in  S_{\alpha} \setminus I$). 
Since $\emptyset$ is independent of ${\cal U}_\alpha$, such an $I$ must exist. 

\begin{lemma}\label{fcp:cl}
For a posimodular function $f:2^V \to \{0,1,\dots d\}$, let  
$\alpha$ be a satisfiable  assignment of $\varphi_f(\overline{x})$. 
 Let $I$ be defined as above.
Then we have $S_\alpha={\sc  FCP}(\varphi_f(\overline{x}); I)$.
\end{lemma}

\begin{proof}
If $I= S_{\alpha}$, we have $S_\alpha={\sc  FCP}(\varphi_f(\overline{x}); S_\alpha)$, since $\alpha$ is a satifiable assignment of $\varphi_f(\overline{x})$. 
Assume that $S_\alpha \setminus I$ is not empty.
Then for each element $v \in S_\alpha \setminus I$, 
$I \cup \{v\}$
is dependent of ${\cal U}_\alpha$, i.e.,
some $U \in {\cal U}_\alpha$ satisfies $U  \setminus I =\{v\}$.
This implies that $\varphi_f(\overline{x})$ contains a clause $c$ such that 
$P(c)=\{v\}$ and $N(c)=U\setminus \{v\} \,(\subseteq I)$.
Thus {\sc  FCP}$(\varphi_f(\overline{x}); I)$ contains $v$ for all  $v \in S_\alpha \setminus I$, which implies $S_{\alpha} \subseteq {\sc  FCP}(\varphi_f(\overline{x}); I)$. 
Since $I \subseteq S_\alpha$ and $\alpha$ is satisfiable for $\varphi_f(\overline{x})$, 
we have $S_\alpha={\sc  FCP}(\varphi_f(\overline{x}); I)$. 
\end{proof}

\begin{lemma}\label{opt:lemmaa}
For a posimodular function  $f:2^V \to \{0,1,\dots d\}$, it holds that  $|\{ x \in \{0,1\}^n \mid \varphi_f(x)=1\}| \leq \sum_{i=0}^d
 {n \choose i}\,(=O(n^d))$.
\end{lemma}

\begin{proof}
By Lemma \ref{fcp:cl}, for each satisfiable assignment $\overline{\alpha}$ of $\varphi_f$, we have an independent set $I$ of ${\cal U}_\alpha$ such that 
 $S_\alpha={\sc  FCP}(\varphi_f(\overline{x}); I)$. 
Since $I$ is also independent of ${\cal U}$, $|I|\leq d$ holds by Lemma \ref{unreach2:lem}, which completes the proof. 
\end{proof}

\begin{remark}
\label{remark1}
\rm
 Lemma \ref{fcp:cl} indicates that 
Step 2 of {\sc MinPosimodular($f$)} can be executed by applying
 {\sc  FCP} for all subsets $T$ of $V$ with
$|T| \leq d$.
For each $T$, ${\sc  FCP}(\varphi_f(\overline{x}); T)$ can be computed from $\cal U$ in $O(d |{\cal U}|)=O(n^{d+1})$ time.
Thus, after computing $\cal U$, 
Step 2 of {\sc MinPosimodular($f$)} can be implemented to run in 
$O(n^d(n^{d+1}+T_f))=O(n^{2d+1}+n^dT_f)$ time.
\end{remark}

\nop{
\begin{remark}
\rm
The number of oracle calls of
{\sc MinPosimodular($f$)} is $O(n^d)$.
The algorithm needs to find all minimal unreachable sets $U \in {\cal U}$.
By Lemma~\ref{unreach2:lem}(i), all  $U \in {\cal U}$ satisfy
$|U|\leq d+1$. 
However, 
 we can check whether each subset $U'$ with $|U'|=d+1$
is a minimal unreachable set or not by checking whether all subset $U''$
of $U'$ with $|U''|=d$ is reachable; we need not 
query the value of $f(U)$ for any $U \subseteq V$ with $|U|=d+1$. 
\end{remark}
}
Summarizing  the  arguments given so far, we have the following theorem.

\begin{theorem}
For general $d$, the posimodular function minimization  
 can be solved in $O(n^dT_f+n^{2d+1})$ time.
\end{theorem}

\begin{proof}
Let us analyze the complexity of {\sc MinPosimodular($f$)}. 
Clearly, Steps 1 and 3 can be executed in $O(nT_f)$ and $O(n)$ time, respectively. 
As for  Step 2, $\cal U$ can be computed in $O(n^dT_f+n^{d+1})$ time. 
Here we remark that it is not necessary to  query the value of $f(U)$ for any $U \subseteq V$ with $|U|=d+1$, if we know  $f(W)$ for all $W \subseteq V$ with $|W| \leq d$. 
This together with Remark \ref{remark1} implies that 
Step 2 requires $O(n^dT_f+n^{2d+1})$ time.

Therefore, in total,  {\sc MinPosimodular($f$)} requires $O(n^dT_f+n^{2d+1})$ time. 
\end{proof}

\subsection{Corollaries of our algorithmical results}\label{byproduct:subsec}

Let us first consider generating all minimizers of a posimodular function $f:2^V\to \{0,1,\dots d\}$. 
Note that $f$ might have exponentially many minimizers. 
In fact, if $f=0$, then we have $2^n-1$ minimizers.
We thus consider output sensitive algorithms for it. 

It follows from Lemma \ref{opt:lem}
 that {\sc MinPosimodular($f$)} finds all locally minimal minimizers of $f$.
Let $S$ be a minimizer of $f$ which is not locally minimal. 
By definition of locally minimality, there exists a chain $T_0\,(=T)
 \subsetneq T_1 \subsetneq \dots \subsetneq T_k\,(=S)$ from some locally
 minimal minimizer $T$ of $f$ such that for all $i=1, \dots ,k$, $T_i$
 is a minimizer of $f$ and $|T_{i}\setminus T_{i-1}|=1$. 
Therefore, after generating all locally minimal minimizers of $f$, we check whether 
$T \cup \{v\}$ is a minimizer of $f$ for each minimizer $T$ of $f$ and $v \not\in T$.
This implies that  all (not only locally minimal) minimizers of $f$
can be generated in $O(nT_f)$ delay after applying
{\sc MinPosimodular($f$)} once.

\begin{corollary}
For a posimodular function $f:2^V\to \{0,1,\dots, d\}$, 
we can generate all minimizers of $f$ in $O(nT_f)$ delay, 
after $O(n^dT_f+n^{2d+1})$ time to compute the first minimizer of $f$.  
\end{corollary}

We next show that
the family ${\cal X}(f)$ of all extreme sets
can be obtained  as an application of {\sc MinPosimodular}.

Recall that 
a subset $X$ of $V$ is called {\em extreme} if
every nonempty proper subset $Y$ of $X$ satisfies $f(Y)> f(X)$. 
By definition, ${\cal X}(f)$ contains all singletons $\{v\}$, $v \in V$, 
and any extreme set $X$ with $|X|\geq 2$ is locally minimal.
This together with Lemma \ref{opt:lem}  implies that Algorithm {\sc MinPosimodular} checks all possible candidates for extreme sets.  
By the following simple observation,  we only check the extremeness among such candidates.

\begin{lemma}
If a family  $\cal Q  \subseteq 2^V$ contains all extreme sets of $f$, 
then $X \in Q$ is extreme for $f$ if and only if  any nonempty proper subset $Y$ of $X$ with $Y  \in {\cal Q}$ satisfies 
$f(Y) > f(X)$.  
\end{lemma}

\begin{proof}
If some nonempty proper subset $Y$ of $X$ with $Y  \in {\cal Q}$ satisfies 
$f(Y) \leq f(X)$, then $X$ is not extreme for $f$. On the other hand, if $X$ is not extreme, then some nonempty proper subset $Y$ of $X$ satisfies $f(Y) \leq f(X)$. If $Y$ is not contained in ${\cal Q}$, 
then  $Y$ is not extreme for $f$, and hence there exists an extreme set $Z$ of $f$ 
such that $Z \subseteq Y$ and $f(Z)\leq f(Y)$. 
Note that this $Z$ is a nonempty proper subset of $X$ with  $f(Z) \leq f(X)$, 
which is contained in ${\cal Q}$.
\end{proof}

\noindent
{\bf Algorithm}  {\sc  ComputeExtremeSets($f$)}

\smallskip

\noindent
{\bf Step 1}. Let ${\cal X}:=\emptyset$ and let ${\cal Q}:=\{v \mid v \in V\} \cup\{ V \setminus {\sc  FCP}(\varphi_f(\overline{x}); I) \mid I \subseteq V, |I|\leq d \}$. 

\hspace{.92cm}/* Here all $f(X)$, $X \in {\cal Q}$ are assumed to be stored.  */ 

\smallskip 

\noindent
{\bf Step 2}.  For each $X \in {\cal Q}$ do 

\hspace*{1.49cm} If all nonempty $Y \in {\cal Q}$ with $Y \subsetneq X$
satisfy $f(Y) > f(X)$, then ${\cal X}:={\cal X} \cup \{X\}$. 

\hspace{.92cm}Output ${\cal X}$ (as ${\cal X}(f)$) and halt. \qed

\med

\noindent
Similarly to Algorithm {\sc MinPosimodular}, Step 1 requires 
  $O(n^{d}T_f+n^{2d+1})$ time. 
Moreover, by $|{\cal Q}|= O(n^d)$, Step 2 can be executed in  $O(n^{2d+1})$ time. 

In summary, we have the following result.

\begin{corollary}
For a posimodular function $f:2^V\to \{0,1,\dots, d\}$, 
we can compute the family ${\cal X}(f)$ of all extreme sets of $f$ in $O(n^{d}T_f+n^{2d+1})$ time.
\end{corollary}

\section{Posimodular function maximization}\label{max-sec}
\setcounter{equation}{0}
In this section, 
we consider the posimodular function maximization defined as follows. 

\begin{equation}\label{posi2:prob}
\begin{array}{ll}
\multicolumn{2}{l}{\hspace*{-.10cm} 
{\mbox {\sc Posimodular Function Maximization}}}
\\[.08cm] 
\hspace*{-.10cm}
\mbox{Input:}& \mbox{A posimodular function}\,\,f: 2^V \to \mathbb{R}_+,
 \\[.08cm]
\hspace*{-.10cm}
\mbox{Output:}& \mbox{A nonempty subset }X\,\,\mbox{of }V \mbox{
maximizing }f.
\end{array} 
\end{equation}

\noindent
Here we assume that the optimal value $f(X^*)$ is also output.  
Similarly to the posimodular function minimization, the problem (\ref{posi2:prob}) is in general intractable.

\begin{theorem}\label{max:th}
Any algorithm for the posimodular function maximization requires 
at least $2^{n-1}$ oracle calls.
\end{theorem}
\begin{proof}
Let us first consider the case in which $n$ is even, i.e.,   $n=2k$ for some positive integer $k$. 
Let $g :2^V \to \mathbb{R}_+$ be a function defined by
$g(X)=|X|$ if $|X| \leq k-1$, and $g(X)=k$ otherwise,
and  for a subset $S \subseteq V$ with $|S| \geq k$,
define a function  $g_S :2^V \to 
\mathbb{R}_+$ by
$g_S(X)=g(X)$ if $X\not=S$, and $g_S(X)=k+1$ if $X=S$.
Since $g$ is monotone,  it is  posimodular.
We claim that  $g_S$ is also posimodular.

Note that $g_S(Z)\geq g(Z')$ holds for any pair of subsets $Z$ and $Z'$ with $Z \supseteq  Z'$ except for $Z'=S$. 
Let $X$ and $Y$ be two subsets of $V$ with $X\cap Y \neq \emptyset$. 
In order to check the posimodular inequality  (\ref{posi:eq}), we can assume  
that $S=X  \setminus Y$ or $Y \setminus X$, since all the other cases can be proven easily. 
By symmetry, let  $S=X\setminus Y$.
Then we have $g_S(X)=k$, $g_S(X\setminus Y)=k+1$, and since $|Y \setminus X| \leq n-k-1 = k-1$, 
$g_S(Y) >g_S(Y\setminus X)$ holds.
These imply the posimodular inequality.

Let $q=\sum_{i=k}^n{n \choose i} ~(\geq 2^{n-1})$.
Assume that there exists an algorithm $A$ for the posimodular
function maximization
which requires oracle calls smaller than $q$. 
Let $\cal X$ denote the family of subsets of $V$ which are called by $A$ if a posimodular function $g$ is an input of $A$. 
Since $|{\cal X}| \leq q-1$, we have a subset $S$ such that $S \not\in {\cal X}$ and $|S|\geq k$. 
This implies that $g_S(X)=g(X)$ for all $X \in {\cal X}$, which contradicts that Algorithm $A$ distinguishes between $g$ and $g_S$ (i.e., $A$ cannot know if the optimal value is either $k$ or $k+1$).

Next let us  consider the case in which $n$ is odd, i.e.,   $n=2k+1$ for some nonnegative integer $k$. 
Let $g :2^V \to \mathbb{R}_+$ be a function defined by
$g(X)=|X|$ if $|X| \leq k$, and $g(X)=k+1$ otherwise,
and  for a subset $S \subseteq V$ with $|S| \geq k+1$,
define a function  $g_S :2^V \to 
\mathbb{R}_+$ by
$g_S(X)=g(X)$ if $X\not=S$, and $g(X)=k+2$ if $X=S$.
In a similar way to the previous case, 
we can observe that
at least 
$\sum_{i=k+1}^n
 {n \choose i} \geq 2^{n-1}$ oracle calls are required to solve the posimodular function maximization. 
\end{proof}

\nop{
If we make use of posimodular functions $c \cdot g$ and $c \cdot g_S$ for some positive $c$, instead of $g$ and $g_S$, 
we can derive the following inapproximable result. 
\begin{corollary}
For any positive $c$, any $c$-approximation algorithm for the posimodular function maximization  requires 
at least $2^{n-1}$  oracle calls.
 \end{corollary}
}

Next consider the case where $f: 2^V \to \{0,1,\ldots,d\}$
for a nonnegative integer $d$.
Then we have the following tight result for the posimodular function maximization.  

\begin{theorem}\label{max2:th}
The posimodular function maximization for $f: 2^V \to \{0,1,\ldots,d\}$ with a constant $d$
can be solved in $\Theta(n^{d-1}T_f)$ time. 
\end{theorem}

The following lemma shows the lower bound for the posimodular function maximization, 
where the upper bound will be shown in the next subsection. 

\begin{lemma}
\label{lemma-max-1}
The posimodular function maximization for $f: 2^V \to \{0,1,\ldots,d\}$ requires 
 $\Omega(n^{d-1})$ oracle calls, if $n \geq 2d-2$. 
\end{lemma}

\begin{proof}
Let $g :2^V \to \{0,1,\ldots,d\}$ be a function defined by
$g(X)=|X|$ if $|X| \leq d-2$, and $g(X)=d-1$ otherwise.
For a subset $S \subseteq V$ with $|S| \geq n-d+1~(\geq d-1)$,
define a function  $g_S :2^V \to 
\{0,1,\ldots,d\}$ by
$g_S(X)=g(X)$ if $X\not=S$, and $g_S(X)=d$ if $X=S$.
Since $g$ is monotone, it is posimodular.
We claim that  $g_S$ is also posimodular.

Note that $g_S(Z)\geq g(Z')$ holds for any pair of subsets $Z$ and $Z'$ with $Z \supseteq  Z'$ except for $Z'=S$. 
Let $X$ and $Y$ be two subsets of $V$ with $X\cap Y \neq \emptyset$. 
In order to check the posimodular inequality  (\ref{posi:eq}), we can assume  
that $S=X  \setminus Y$ or $Y \setminus X$, since all the other cases can be proven easily. 
By symmetry, let  $S=X\setminus Y$.
Then we have $g_S(X)=d-1$, $g_S(X\setminus Y)=d$, and since $|Y \setminus X| \leq n-|S|-1 \leq d-2$ and $|Y| > |Y \setminus X|$,  
$g_S(Y) >g_S(Y\setminus X)$ holds.
These imply the posimodular inequality.

In a similar way to the proof of Theorem~\ref{max:th},
we can observe that  the posimodular function maximization
requires
at least 
$\sum_{i=n-d+1}^n
 {n \choose i} = \Omega(n^{d-1})$
 oracle calls, to distinguish among $g$ and all $g_S$ with $|S|\geq n-d+1$.   
\end{proof}

\subsection{Polynomial time algorithm for a constant $d$}

In this section, we present an $O(n^{d-1}T_f)$-time algorithm for
the posimodular function maximization for a constant $d$. 

The following simple lemma implies that the problem can be solved in $O(n^{d}T_f)$ time.

\begin{lemma}\label{max1:lem}
Let $f:2^V\to \{0, 1,\dots , d\}$ be a posimodular function, 
and let $S$ be a maximal maximizer of $f$ $($i.e., a maximizer such that no proper superset is a maximizer of $f$$)$.
Then, $f(X \cup \{v\})\geq f(X)+1$ holds for any pair of a set  $X \subseteq V$ and an  element $v \in V$ such that  $X$, $\{v\}$ and $S$ are pairwise disjoint. 
\end{lemma}
\begin{proof}
By (\ref{posi:eq}), we have
$f(X \cup \{v\})+f(S \cup \{v\})\geq f(X)+f(S)$.
By the maximality of $S$, we have $f(S\cup \{v\})<f(S)$.
Hence, we have $f(X \cup \{v\})>f(X)$.
\end{proof}

\begin{corollary}\label{cor-max1}
Let $f:2^V\to \{0, 1,\dots , d\}$ be a posimodular funiction.
Then we have $|S| \geq n-d$ for any  maximal maximizer $S$ of $f$. 
\end{corollary}

\begin{proof}
Let $k=|S|$, and  let $X_0\,(=\emptyset) \subseteq X_1 \subseteq \dots  \subseteq X_{n-k}\,(=V\setminus S)$ be a chain with $|X_i|=i$ for all $i$. 
Then it follows from Lemma \ref{max1:lem} that 
\begin{equation}
\label{eq-max--00}
f(X_0)\,(=0) < f(X_1) < \dots  < f(X_{n-k}) \,(\leq d), 
\end{equation}
which implies that $n-k \leq d$. 
\end{proof}

\noindent
By the corollary, the posimodular function maximization can be solved in $O(n^{d}T_f)$ time by checking all subsets $X$ with $|X|\geq n-d$.

In the remaining part of this section, we reduce the complexity to $O(n^{d-1}T_f)$ by showing a series of lemmas which assumes that no maximizer of $f$ has size at least $n-d+1$, i.e., by Corollary \ref{cor-max1} and (\ref{eq-max--00}),
\begin{equation}\label{assumption:eq}
\mbox{any maximal maximizer $X^*$ of $f$ satisfies $|X^*|=n-d$ and $f(X^*)=d$.} 
\end{equation}
By (\ref{eq-max--00}), it implies that $n\geq 2d$.

\begin{lemma}\label{max2:lem}
Under the assumption $(\ref{assumption:eq})$, we have the following two statements. 
\begin{description}
\setlength{\parskip}{0cm} 
\setlength{\itemindent}{-0.7cm}
\item[$(${\it i}$)$] For any  maximizer  $S$ of $f$, there exists a maximizer
$S'$ of $f$ with $S' \cap S=\emptyset$ and $|S'|=d$.
\item[$(${\it ii}$)$] 
 Let $S_1,S_2$ be two maximizers of $f$ with $S_1 \cap S_2=\emptyset$.
 Then, there exist two maximizers $X_1,X_2$ of $f$
with $|X_1|=|X_2|=d$ and $X_i \subseteq S_i$, $i=1,2$.
Moreover, any subset $Y \subseteq V$ with 
$X_1 \subseteq Y \subseteq V  \setminus X_2$ or
$X_2 \subseteq Y \subseteq V  \setminus X_1$ is a maximizer of $f$. 
\end{description}
\end{lemma}

\begin{proof}
({\it i}).  Let $S$ be an arbitrary  maximizer of $f$,
and $S_1$ be a maximal maximizer of $f$  with $S_1 \supseteq S$.
By 
 (\ref{assumption:eq}), we have $|S_1| = n-d$ and
hence $|V \setminus S_1|= d$.
It follows from (\ref{eq-max--00}) that $f(V \setminus S_1)=d$, which means that 
$V \setminus S_1$ is a maximizer of $f$ with size $d$ which is disjoint from $S$. 

({\it ii}).
Since
we have $f(V \setminus S_1)+f(V \setminus S_2) \geq f(S_1)+f(S_2)$
by (\ref{posi:eq}),
both $V \setminus S_1$ and $V \setminus S_2$ are also maximizers of $f$.
By  (\ref{assumption:eq}), we have $|V \setminus S_j|\leq n-d$ and $|S_j|\geq d$
for $j=1,2$.
By applying ({\it i}) to $V\setminus S_{j}$ ($j=1,2$), we obtain 
a maximizer $X_j \subseteq S_{j}$ with $|X_{j}|=d$. 
Here we note that $X_1 \cap X_2 =\emptyset$.
Moreover, 
for any set $Z \subseteq V \setminus (X_1 \cup X_2)$,
both $X_1 \cup Z$ and $X_2 \cup Z$ are also maximizers of $f$,
since
we have $f(X_1 \cup Z)+f(X_2 \cup Z)\geq f(X_1)+f(X_2)$
by (\ref{posi:eq}). 
This completes the proof. 
\end{proof}

\noindent

\begin{lemma}\label{max2-2:lem}
Assume that $(\ref{assumption:eq})$ holds.
Let $S$ be
 a maximizer  of $f$ with size $d$,  and
let $X$ be
 a subset of $V$ such that $|X|=f(X)=d-1$ 
 and $X \cap S=\emptyset$.
Then, 
there exists an element $v \in V \setminus (S\cup X)$ with
$f(X \cup \{v\})=d$.
\end{lemma}

\begin{proof}
Let $S'$ be a maximizer of $f$
with $S' \cap S=\emptyset$ and $|S'|=d$ such that
$|S' \setminus X|$ is the minimum.
We note that such an  $S'$ always exists by Lemma~\ref{max2:lem} ({\it i}), and 
 $S' \setminus X \neq \emptyset$ is satisfied by $|S'|>|X|$. 
Moreover, it follows from  Lemma~\ref{max2:lem} ({\it ii}) that 
$V \setminus (X\cup S')$ is also a maximizer of $f$.  
For $v \in S' \setminus X$,
we have
$f(X \cup \{v\})+f(V \setminus(X \cup (S'\setminus \{v\})))$
$\geq f(X)+f(V \setminus(X \cup S'))=2d-1$ 
by (\ref{posi:eq}). 
Therefore, it suffices to show that
$f(V \setminus(X \cup (S'\setminus
 \{v\})))\leq d-1$ to prove $f(X \cup \{v\})=d$. 

Assume to the contrary that $f(V \setminus(X \cup (S'\setminus
 \{v\})))=d$.
By Lemma~\ref{max2:lem} ({\it i}), there exists a maximizer  $S''$ of $f$
with $|S''|=d$ and
 $S'' \cap (V \setminus(X \cup (S'\setminus
 \{v\})))=\emptyset$, i.e., $S'' \subseteq X \cup (S'\setminus
 \{v\})$, which contradicts the minimality of $|S' \setminus X|$.
\end{proof}

\noindent
We remark that $S$ and $X$ in Lemma \ref{max2-2:lem}
always exist if  $(\ref{assumption:eq})$ is satisfied. 
In fact, by
Lemma \ref{max2:lem}, we have two maximizers $X_1$ and $X_2$ of $f$ such that 
$|X_1|=|X_2|=d$, $X_1 \cap X_2=\emptyset$, and $V\setminus X_2$ is also a maximizer of $f$.    
Let $S=X_1$ and $X=X_2 \setminus \{v\}$ for any $v \in X_2$.
Then $S$ satisfies the condition  in Lemma \ref{max2-2:lem}, and 
since $V\setminus X_2$ is a maximal maximizer of $f$,  
 (\ref{eq-max--00}) implies that $X$ also satisfies the condition in Lemma \ref{max2-2:lem}. 

\begin{lemma}\label{max2-3:lem}
 Let ${\cal X}$ be the family of all subsets $X$ of $V$ 
such that $|X|=d-1$ and $X \cap S\neq \emptyset$ for all maximizers $S$ of $f$ with $|S|=d$.
Then, under the assumption
$(\ref{assumption:eq})$,
 we have $|{\cal X}| =O(n^{d-3})$.
\end{lemma}
\begin{proof}
 By Lemma~\ref{max2:lem}, there exist two maximizers $S_1$ and $S_2$ of $f$
with $|S_1|=|S_2|=d$ and $S_1 \cap S_2 =\emptyset$.
Clearly, $|{\cal X}|$ is bounded by
the number of sets $X$ with size $d-1$ with
$X \cap S_1, X \cap S_2 \not=\emptyset$,
which is 
\[\sum_{i,j>0, i+j\leq d-1}{d \choose i}{d \choose j}{n-2d
 \choose d-1-i-j}\leq \sum_{k=2}^{d-1}{2d \choose k}
{n-2d \choose d-1-k}
=O(n^{d-3}). \]
\end{proof}

Let $c$ be a constant such that $|{\cal X}| \leq cn^{d-3}$ for 
${\cal X}$ in Lemma \ref{max2-3:lem}.  
Based on these lemmas, we can find a maximizer of $f$ in the following manner:

\med

\noindent
{\bf Algorithm}  {\sc  MaxPosimodular($f$)}

\smallskip

\noindent
{\bf Step 1}.  Find a subset $X_1$ of $V$ such that $|X_1|\geq n-d+1$ and
 $f(X_1)=\max\{f(X) \mid X\subseteq V, |X| \geq n-d+1 \}$. 
If $f(X_1)=d$, then output $X_1$ and halt.

\smallskip

\noindent
{\bf Step 2}. Find a subset $X_2$ of $V$ such that $|X_2|= d-1$ and
 $f(X_2)=\max\{f(X) \mid X\subseteq V, |X| = d-1 \}$. 
If $f(X_2)= d$, then output $X_2$ and halt. 
If $f(X_2)\leq d-2$, then output $X_1$ and halt.

\smallskip

\noindent
{\bf Step 3}. 
Choose $\min\{cn^{d-3}+1, |{\cal X}_1|\}$ members $X$ from 
  ${\cal X}_1=\{X \subseteq V \mid |X|=d-1, f(X)=d-1\}$. 
For each such $X$, if 
$f(X \cup \{v\})=d$ for some $v \not\in X$, 
then output $X \cup \{v\}$ and halt. 

\smallskip

\noindent
{\bf Step 4}.  Output $X_1$  and halt.

\begin{lemma}\label{posimax:lem}
Algorithm {\sc  MaxPosimodular($f$)} solves the posimodular function minimization for $f:2^V\to \{0,1,\dots , d\}$ for a constant $d$ in $O(n^{d-1}T_f)$ time. 
\end{lemma}

\begin{proof}
Let us first prove the correctness of the algorithm. 
Let $S$ be a maximal maximizer of $f$.
Assume that $f(S)=d$ holds. 
Then Corollary \ref{cor-max1} implies that 
$|S|\geq n-d$. 
If $|S|\geq n-d+1$, then $S$ can be found in Step 1.
On the other hand, if $|S|= n-d$,  then 
we have (\ref{assumption:eq}). 
 By the discussion after Lemma \ref{max2-2:lem}, $f(X_2)\geq d-1$ must hold.
If $f(X_2)=d$, then $X_2$ is clealy a maximizer of $f$ which is output in Step 2. 
Otherwise (i.e., $f(X_2)=d-1$), 
by Lemma \ref{max2-2:lem} together with the discussion after Lemma \ref{max2-2:lem}, 
for each subset $X$ with $|X|=f(X)=d-1$, we only check  if $f(X \cup \{v\})=d$ for some $v \not\in X$. 
Moreover, it follows from Lemma \ref{max2-3:lem} that we only check 
at most $cn^{d-3}+1$ many such $X$. 
Therefore, in this case, Step 3 correctly outputs a maximizer of $f$. 

Assume next that  $f(S) \leq d-1$. 
Then Algorithm {\sc  MaxPosimodular($f$)}  output $X_1$ in Step 2 or 4, which is correct, since  there exists a maximal maximizer of size at least $n-d+1$ by Corollary \ref{cor-max1}. 

As for the time complexity of Algorithm {\sc  MaxPosimodular($f$)}, we see that 
Steps 1 and 2 can be executed in $O(n^{d-1}T_f)$ time. 
Since Steps 3 and 4 respectively require $O(n^{d-2}T_f)$ and $O(n)$ time, 
in total, algorithm requires $O(n^{d-1}T_f)$ time. 
\end{proof}

\begin{remark}
\rm
For  ${\cal X}$  defined in Lemma~\ref{max2-3:lem},
we have $|{\cal X}|=O(n^{d-2})$
 if $d=O(\sqrt{n})$.
As observed in the proof of Lemma~\ref{posimax:lem}, 
the time complexity of Algorithm {\sc  MaxPosimodular($f$)}
is $O((n^{d-1}+n|{\cal X}|)T_f)$.
Hence, it follows that the posimodular function maximization
has
 time complexity  $\Theta(n^{d-1}T_f)$ even for  $d=O(\sqrt{n})$.
\end{remark}

\med
\med

\noi
 {\bf Acknowledgments:}
We would like to express our thanks to
S. Fujishige, M. Gr{\"{o}}tschel, and S. Tanigawa
 for their helpful comments.
This research was partially supported by the Scientific
Grant-in-Aid   from  Ministry of Education, Culture, 
Sports, Science and Technology of Japan.

\bibliographystyle{siam}    

\bibliography{./ishii_ref}

\end{document}